\documentclass[hidelinks,onefignum,onetabnum]{siamart220329}

\usepackage{amsfonts}
\usepackage{graphicx}
\usepackage{amsopn}
\ifpdf
  \DeclareGraphicsExtensions{.eps,.pdf,.png,.jpg}
\else
  \DeclareGraphicsExtensions{.eps}
\fi

\usepackage{graphicx,epstopdf} 
\usepackage[caption=false]{subfig}

\newsiamremark{remark}{Remark}
\newsiamremark{hypothesis}{Hypothesis}
\crefname{hypothesis}{Hypothesis}{Hypotheses}
\newsiamthm{claim}{Claim}

\headers{Bounding the approach to oligarchy}{David W. Cohen and Bruce M. Boghosian}

\title{Bounding the approach to oligarchy in a variant of the yard-sale model}

\author{David W. Cohen\thanks{Department of Mathematics, Tufts University, Medford, MA 
  (\email{david.cohen@tufts.edu}, \url{https://sites.tufts.edu/davidcohen/}).\funding{The first author acknowledges the support of the NDSEG fellowship.}} \and Bruce M. Boghosian\thanks{Department of Mathematics, Tufts University, Medford, MA  (\email{bruce.boghosian@tufts.edu}).}}

\newcommand{\ysmmin}[2]{\left(#1\wedge#2\right)}

\newcommand{\ddw}[1]{\frac{\partial #1}{\partial w}}
\newcommand{\twoddw}[1]{\frac{\partial^2 #1}{\partial w^2}}

\newcommand{\ddt}[1]{\frac{\partial #1}{\partial t}}

\newcommand{\totdd}[2]{\frac{d #1}{d #2}}
\newcommand{\tottwodd}[2]{\frac{d^2 #1}{d {#2}^2}}
\newcommand{\frechet}[2]{\frac{\delta #1}{\delta #2}}

\newcommand{\Ddiff}[1]{D\left[w, #1\right]}

\newcommand{\ysm}{Yard-Sale Model}

\ifpdf
\hypersetup{
  pdftitle={Bounding the approach to oligarchy in a variant of the \ysm},
  pdfauthor={David W. Cohen \and Bruce M. Boghosian}
}
\fi

\begin{document}

\maketitle

\begin{abstract}
  We present analytical results for the Gini coefficient of economic inequality under the dynamics of a modified Yard-Sale Model of kinetic asset exchange. A variant of the Yard-Sale Model is introduced by modifying the underlying binary transaction of the classical system. It is shown that the Gini coefficient is monotone under the resulting dynamics but the approach to oligarchy, as measured by the Gini index, can be bounded by a first-order differential inequality used in conjunction with the differential Gr\"onwall inequality. This result is in the spirit of entropy -- entropy production inequalities for diffusive PDE. The asymptotics of the modified system, with a redistributive tax, are derived and shown to agree with the original, taxed Yard-Sale Model, which implies the modified system is as suitable for matching real wealth distributions. The Gini -- Gini production inequality is shown to hold for a broader class of models.
\end{abstract}

\begin{keywords}
  yard-sale model, econophysics, non-linear Fokker-Planck equation, mean-field theory, McKean-Vlasov equations
\end{keywords}

\begin{AMS}
  91B80, 82C22, 82C31
\end{AMS}

\section{Introduction}
The \ysm{} is a well-studied model of kinetic asset exchange introduced by A. Chakraborti and named by B. Hayes \cite{AC2002, BH2002}. At its core, the \ysm{} is a specification of how a wealth transaction is to be conducted between two agents randomly selected from a population. From this point, the system may be altered and then studied as a stochastic finite-agent system or, through the use of thermodynamic limits and other techniques from mathematical physics, as a deterministic, continuum equation of motion for the distribution of wealth.

The classical \ysm{} transaction is stochastic: At each integer time $t$, two agents from an $N$ agent population are selected at random without replacement and their wealths are updated according to the rule \begin{equation}\label{eq:ysmMicroTransact}
\begin{pmatrix}w_{t+1}^i\\w_{t+1}^j\end{pmatrix} = \begin{pmatrix}w_{t}^i\\w_{t}^j\end{pmatrix} + \sqrt{\gamma} \ysmmin{w_t^i}{w_t^j} \begin{pmatrix}1\\-1\end{pmatrix}\eta,
\end{equation} where $\gamma\in(0,1)$ is a transaction intensity parameter, $\eta$ is a random variable with outcomes $-1$ and $+1$ with equal probability, and $\wedge$ is the $\min$ operator.

Under this rule the expected change of wealth for each agent is zero yet it is well known that wealth condenses and an oligarchy forms as time progresses. The definition of the classical transaction rule is premised on the belief that the poorer agent's wealth should determine the magnitude of a binary transaction and that the changes in wealth are associated to a lack of perfect information in a given transaction.\footnote{Informally, in a transaction with a very wealthy agent, a poorer agent is willing to stake a small fraction of their wealth but cannot possibly stake the same fraction of the other's wealth as losing would horrifically bankrupt the poorer agent.} Adhering to this principle, the \ysm{} is the simplest, non-trivial stochastic transaction that cannot send an agent to negative wealth when all agents are initialized with positive wealth.

In \cite{BMB2014b} a time step is introduced into the transaction \cref{eq:ysmMicroTransact} to give an infinitesimal characterization of the process by a limiting procedure involving the number of agents and time step combined with an analogue of the molecular chaos assumption (\textit{Stosszahlansatz}) -- termed the random-agent approximation. This procedure produces a deterministic continuum equation. The equation of motion for the probability distribution of agents in wealth-space under the classical \ysm{} dynamics of \cref{eq:ysmMicroTransact} is \begin{equation}\label{eq:ysm}\ddt{\rho(w,t)}= \twoddw{}\left[\left(\frac{\gamma}{2}\int_0^\infty\,dx\,\ysmmin{w}{x}^2\rho(x,t)\right)\rho(w,t)\right].\end{equation}

A frequently studied summary statistic of distributions of wealth is the Gini coefficient of economic inequality, which was introduced by Corrado Gini in 1912 \cite{CG1912,MR3012052}. The Gini coefficient maps a wealth distribution to a value in $[0,1]$. Values near zero correspond to more egalitarian distributions, whereas values of the Gini approaching unity indicate inequality and oligarchy. If $\mu$ is the mean wealth of a distribution of wealth with density $\rho$ and $X,Y$ are i.i.d. random variables with density $\rho$ then \[G[\rho]:=\frac{\mathbb{E}\left[|X-Y|\right]}{2\mu}.\]

A Lorenz curve represents a distribution of wealth in the unit square, $[0,1]\times [0,1]$, by plotting on the abscissa the fraction of a population with wealth less than $w$ and the fraction of total wealth held by this subset of the population on the ordinate. More precisely, the Lorenz curve is a $w$-parameterized plot of \[\left(F(w),L(w)\right) =  \left(\frac{\int_0^w\,dy\,\rho(y) }{\int_0^\infty\,dy\,\rho(y)},\frac{\int_0^w\,dy\,y\rho(y) }{\int_0^\infty\,dy\,y\rho(y)}\right)\] The Gini coefficient also has a geometric interpretation when the Lorenz curve is used to represent a distribution of wealth. If the population were to all have the mean wealth, then the Lorenz curve would be the identity and correspond to an egalitarian society. As a population moves toward total oligarchy, the Lorenz curve is pushed into the bottom right corner of the unit square. The Gini coefficient can be equivalently defined as twice the area between the diagonal and the Lorenz curve of a distribution of wealth -- this varies between zero and unity. For the functions $F$ and $L$ defined above in terms of $\rho$, we have \[G[\rho] = 2\int_0^1\,ds\, \left(s-(L\circ F^{-1})(s)\right).\]

Boghosian et al. in \cite{BMB2014a} established that the Gini coefficient is monotone under the dynamics of the continuum model \cref{eq:ysm} of the classical \ysm{}. 
This result was shown both for the master equation and the resulting non-linear partial integro-differential Fokker-Planck equation first derived in \cite{BMB2014b}. Thus the Gini coefficient is a Lyapunov functional \cite{TDF2010, ANM2015} for the \ysm. Chorro showed via a martingale convergence theorem argument in \cite{CC2016} that the finite-agent system likewise approaches oligarchy (in a probabilistic sense). B\"orgers and Greengard in \cite{CB2023} produced yet simpler proofs of wealth condensation for the classical finite-agent system. B\"orgers and Greengard \cite{CB2023} and Boghosian \cite{BMB2017} have similar results for transactions that are biased in favor of the wealthier agent. Beyond the results for the \ysm{} and its variants, Cardoso et al. in \cite{BHFC2023, BHFC2021} and F. Cao and S. Motsch in \cite{MR4619925} have shown that wealth condensation is more likely the rule than the exception for more general unbiased binary exchanges.

Despite the many methods of showing that wealth condensation occurs, we are not aware of any explicit bounds on the rate of increase of the Gini coefficient under models for which the Gini coefficient increases monotonically.

Cao in \cite{FC2023} investigated the rate of change of the Gini coefficient under the dynamics of the repeated averaging model (sometimes called the divorce model) where the binary transaction sends each agent to the mean wealth of the two transacting agents. Cao mentions that in ``econophysics literature, analytical results on [the] Gini index are comparatively rare.''

Here we produce an analytic result on the Gini coefficient for a modified, yet reasonable, \ysm{} that steps beyond the standard propositions of monotonicity. In particular, the result is a bound on the rate of production of inequality, as measured by the Gini coefficient, under the modified dynamics. From a physical perspective, this is akin to finding a bound on the rate of entropy production while still having a thermodynamic second law.

That is to say, the Gini -- Gini production relation derived below is similar to entropy -- entropy production inequalities for diffusive PDE. For the systems examined in \cite{MR3497125}, the entropy production bound gives a minimum rate of decrease for the relative entropy. Whereas here the Gini -- Gini production bound gives a maximum rate of increase for the measure of economic inequality.

Importantly, the main theorem generalizes to a broader class of models than just those in which the microtransaction permits a discrete set of outcomes.

The paper is organized as follows. The variant of the \ysm{} on which the present paper focuses is motivated and defined in \cref{sec:mysm}. The Gini coefficient is briefly reviewed in \cref{sec:gini} with particular focus on its invariance under a normalization of the equations of motion. In \cref{sec:giniTime} it is proven both that the Gini coefficient increases monotonically in time under the induced dynamics and that its rate of increase may be bounded. This result is then re-stated for a more general class of evolutionary models. The evolutionary, integro-differential PDE are numerically solved to demonstrate the bound holding in experiment. Plots and descriptions of the numerical method are included. The asymptotics of the modified system when a redistributive tax is incorporated are derived in \cref{sec:asymptotics} and shown to match the classical \ysm{} with taxation.

\section{The modified \ysm}
\label{sec:mysm}
Henceforth we sometimes use the abbreviation YSM for the Yard-Sale Model.

Let there be $N>1$ agents each with a dimensionally-meaningful wealth $\theta^i_s$ indexed by $i=1,\ldots,N$ and a time $s\geq0$. Time subscripts $s$ are occasionally omitted for clarity when the time is unimportant for the expression. Let $W_\theta:=\sum \theta^i$ be the total wealth and $\mu_\theta:=W_\theta/N$ be the average wealth per agent. To each of $N$ agents, associate a dimensionless quantity $w^i:=\theta^i/\mu_\theta$ obtained by dividing wealth by mean wealth.\footnote{In what follows, we still call $w$ wealth despite its dimensionless nature.} 

Let $k=0,1,2,\ldots$ and $\Delta t \in (0,1)$. The modified version of the \ysm{} has a binary transaction between agents indexed by $i$ and $j$ given by \begin{equation}\label{eq:mysmMicroTransact}
\begin{pmatrix}w_{(k+1)\Delta t}^i\\w_{(k+1)\Delta t}^j\end{pmatrix} = \begin{pmatrix}w_{k\Delta t}^i\\w_{k \Delta t}^j\end{pmatrix} + \sqrt{\gamma\Delta t} \phi\left(w^i_{k \Delta t},w^j_{k \Delta t}\right) \begin{pmatrix}1\\-1\end{pmatrix}\eta,
\end{equation} 

with \begin{equation}\label{eq:mysmMicroKernel}
\phi\left(w^i,w^j\right) = \begin{cases}
(w^i \wedge w^j) &\text{\quad if $w^i\wedge w^j <1$;} \\
\sqrt{ (w^i \wedge w^j)}&\text{\quad if $w^i\wedge w^j \geq1$,}
\end{cases} 
\end{equation} where $\gamma\in(0,1)$, $\wedge$ is the $\min$ operator, and $\eta$ is a random variable that takes values $-1$ and $+1$ with equal probability. The standard YSM has $w^i \wedge w^j$ outside of the square root function regardless of the size of $w^i \wedge w^j$. The importance of the square root in \cref{eq:mysmMicroKernel} will become clearer below in light of the more general statement of \cref{cor:otherModels} and the relation of the diffusion coefficient kernel to that of the Gini coefficient.

The expectation of the transaction is zero, that is \[\mathbb{E}\left[\begin{pmatrix}w_{(k+1)\Delta t}^i\\w_{(k+1)\Delta t}^j\end{pmatrix}-\begin{pmatrix}w_{k\Delta t}^i\\w_{k \Delta t}^j\end{pmatrix}\right]=0,\] and in this way the process is a martingale. The variance of the transaction is \begin{equation}\label{eq:mysmVar}
    \mathbb{E}\left[\left(\begin{pmatrix}w_{(k+1)\Delta t}^i\\w_{(k+1)\Delta t}^j\end{pmatrix}-\begin{pmatrix}w_{k\Delta t}^i\\w_{k \Delta t}^j\end{pmatrix}\right)^2\right] = \gamma \Delta t \left(\phi\left(w_{k\Delta t}^i,w_{k\Delta t}^j\right)\right)^2
\end{equation}


This modification maintains the important features that the poorer agent's wealth is the determining quantity in the exchange and no exchange can send an agent to negative wealth. Note that the comparison of dimensionless wealth in the piecewise micro-transaction is equivalent to checking if the minimum of the two dimension-full wealths is above or below mean wealth $\mu_{\theta}$.

Under the same assumptions common in mathematical physics that are laid out in \cite{BMB2014c} and \cite{BMB2014b} -- namely assuming independence of the laws of each agent as the number of agents goes to infinity -- the limit as $\Delta t \rightarrow 0$ and $N\rightarrow \infty$ leads to a continuum equation for the law of a single prototypical agent in the population. The nonlinear evolution equation is found via the Kramers-Moyal expansion of the Chapman-Kolmogorov equation after making the random-agent approximation \cite{BMB2014b} and applying Pawula's Theorem. A more formal derivation from the stochastic process would use methods from the propagation of chaos literature on interacting particle systems of the Boltzmann type, see \cite{MR4590314,LPC2022a, LPC2022b, ALS1991} for more details. The yard-sale model and its variants result in McKean-Vlasov stochastic differential equations \cite{VB2015} for which the drift and/or diffusion coefficients are functionals of the dependent variable.

Let $\rho(w,t)$ be the probability density of agents in the dimensionless wealth variable. Since the dimensionless quantity $w$ was obtained by dividing wealth $\theta$ by mean wealth $\mu_\theta$, the first moment of $\rho(w,t)$ is also $1$, that is $\int_{\mathbb{R}^+}w\rho =1$. Therefore, by construction, the mean (dimensionless) wealth is $1$. 

This leads to a Fokker-Planck equation for the normalized agent distribution  \begin{equation}\label{eq:mysm}
\ddt{\rho(w,t)}=\twoddw{}\left[\underbrace{\frac{\gamma}{2}\left(\int_0^\infty\,dx\,\kappa(w,x)\rho(x,t)\right)}_{:=D[w,\rho(\cdot,t)]}\rho(w,t)\right],
\end{equation} where \begin{equation}
\kappa(w,x):=\begin{cases}
(w \wedge x)^2 &\text{\quad if $w\wedge x <1$;} \\
(w \wedge x)&\text{\quad if $w\wedge x \geq 1$.} \end{cases}\end{equation} The diffusion coefficient is proportional to the expected, time-infinitesimal variance experienced by an agent interacting with the mean-field. In light of \cref{eq:mysmVar}, this is why $\kappa$ is seen to be $\phi^2$.


Both the zeroth and first moments of $\rho$ are conserved quantities; the former being the probability mass (implying conservation of number of agents) and the latter corresponding to conservation of total wealth, which we take canonically to be 1.

Throughout this paper we assume that each agent has positive wealth and the support of all distributions are $\mathbb{R}^+$.

\section{The Gini coefficient under normalization}
\label{sec:gini}
Here for the reader's convenience we define the Gini coefficient and show its invariance under a particular transformation of a distribution of wealth.

Let $Q:[0,\infty)\rightarrow [0,\infty)$ have finite zeroth and first moment denoted $N_Q$ and $W_Q$, respectively. Define $\mu_Q = W_Q/N_Q$. The transformation to a distribution with unit zeroth and first moment is given by \[q(w):=\frac{\mu_Q}{N_Q}Q\left(\mu_Qw\right).\] 

$Q$ can be viewed as a wealth distribution of a population of $N_Q$ agents with total population wealth $W_Q$. The Gini coefficient of economic inequality can be expressed as\footnote{This expression is easily derivable from the statement in the introduction of the Gini coefficient as scaled mean absolute deviation by writing out the expectation and using $|a-b| = (a+b) - 2\ysmmin{a}{b}$.}
\begin{align}
G[Q] &= 1 - \frac{1}{N_QW_Q}\int_0^\infty\,dw\,\int_0^\infty\,dx\,\ysmmin{w}{x}Q(w)Q(x) \nonumber\\
&=1 - \int_0^\infty\,dw\,\int_0^\infty\,dx\,\ysmmin{w}{x}q(w)q(x),\label{eq:gini}
\end{align}
which we can equivalently call $G[q]$ under the normalizing transformation $Q\mapsto q$. Thus $G$ is invariant under the transformation between $q$ and $Q$.

The Gini coefficient is $0$ for a density concentrated on mean wealth (that is, for a wealth-egalitarian society) whereas it approaches its upper limit of $1$ as the wealth is concentrated into an ever-vanishing proportion of the population. See \cite{BMBCB2023,ADL2018} for a discussion of the nonstandard properties of wealth distributions that maximize the Gini coefficient under the dynamics of \cref{eq:ysm}.

We will work in the space of normalized wealth distributions in which both the zeroth and first moments are unity.

\begin{lemma} \label{lem:gini} Two derivatives of the Fr\'echet derivative of the Gini coefficient under normalization is twice the density, that is 
    \begin{equation}
        \tottwodd{}{w}\frechet{G}{\rho} = 2\rho(w).
    \end{equation}
\end{lemma}
This is a routine calculation and follows from the symmetry of the kernel in the double integral that defines $G[\rho]$. We state the proof, despite its simplicity, since the result is repeatedly used.
\begin{proof}
Since the integral kernel of $G$ is symmetric in its arguments, the Fr\'echet derivative of $G$ is \[
\frechet{G}{\rho} = -2\int_{\mathbb{R}^+}\,dx\, \ysmmin{w}{x}\rho(x).
\] The first $w$-derivative of the Fr\'echet derivative is
\[
\totdd{}{w}\frechet{G}{\rho} = -2\int_w^\infty \,dx\,\rho(x)
\] and the next $w$-derivative is
\[
\tottwodd{}{w}\frechet{G}{\rho} = 2\rho(w).
\]
\end{proof}

\section{Bounding the rate of increase of the Gini coefficient}
\label{sec:giniTime}
We now turn to the main results of the paper: That the Gini coefficient, despite being monotonically increasing under the modified \ysm{} dynamics, can have a non-trivial bound on its rate of change in time. This bound may be carried over into a bound on the value of the Gini coefficient at a future time.

By $G(t)$, we mean $G[\rho(\cdot,t)]$ where $\rho$ is a solution to \cref{eq:mysm}.

\begin{theorem}[Increasing inequality]\label{thm:GiniIncr}
The Gini coefficient \cref{eq:gini} is monotone increasing under the dynamics of \cref{eq:mysm}.
\end{theorem}
\begin{proof}
By direct calculation, we have that
\begin{align*}
\totdd{G}{t} &= \int_0^\infty\,dw\,\frechet{G}{\rho}\ddt{\rho} \\
&= \int_0^\infty\,dw\,\frechet{G}{\rho}\twoddw{}\left[\Ddiff{\rho(\cdot,t)}\rho(w,t)\right] &&\text{\qquad (by \cref{eq:mysm})} \\
&= \int_0^\infty\,dw\,\left(\twoddw{}\frechet{G}{\rho}\right)\Ddiff{\rho(\cdot,t)}\rho(w,t) &&\text{\qquad (via two IbP)}\\
&= 2\int_0^\infty\,dw\,\Ddiff{\rho(\cdot,t)}\rho(w,t)^2 &&\text{\qquad (by \cref{lem:gini})}\\
&\geq 0 &&\text{\qquad (since $\Ddiff{\rho} \geq 0$).}
\end{align*}
\end{proof}

The above result can be strengthened when there is mass that is not concentrated at the origin.

\begin{corollary}[Strictly increasing inequality]\label{cor:strictlyIncr}
Let $s>0$ and $\epsilon \in (0,1).$ If there exists $a>0$ such that \[\int_a^\infty\,dw\,\rho(w,s)>\epsilon\] then there exists $\delta>0$ such that \[\totdd{G}{t}\bigg|_{t=s}>\delta.\]
\end{corollary}

\begin{proof}
Let $C = \gamma(a^2 \wedge a)$, which is positive. Making use of the calculation in \cref{thm:GiniIncr}, we have
\begin{align*}
\totdd{G}{t}\bigg|_{t=s} &=  2\int_{\mathbb{R}^+} \,dw\, D[w,\rho(\cdot,s)]\left(\rho(w,s)\right)^2\\
&\geq \gamma \int_a^\infty \,dw\,\int_a^\infty \,dx\, \kappa(w,x)\rho(x,s)\left(\rho(w,s)\right)^2\\
&\geq C\epsilon \int_a^\infty \,dw\,\left(\rho(w,s)\right)^2\\
&\geq C\epsilon \frac{\epsilon^2}{4(b-a)} &&(\text{Jensen's inequality})\\
&>0.
\end{align*}
The step prior to the application of Jensen's inequality uses that there must exist $b$ such that $0<a<b<\infty$ for which \[\int_a^b\,dx\,\rho(x,s)>\frac{\epsilon}{2}.\]

Letting \[\delta = \frac{C\epsilon^3}{4(b-a)}>0\] completes the proof.
\end{proof}

\cref{cor:strictlyIncr} implies that $G(t) \rightarrow 1$ as $t\rightarrow \infty$ as the function is bounded and monotone increasing; this also holds for the classical \ysm{}.

\begin{theorem}[Bounding the rate of inequality production]\label{thm:GiniBound}
If the initial datum $\rho_0(w) = \rho(w,0)$ is in $L^\infty(\mathbb{R}^+)$ and the dynamics keep $\rho(w,t)$ in $L^\infty(\mathbb{R}^+)$ up to time $s>0$ then the rate of change of the Gini coefficient \cref{eq:gini} is bounded by \begin{equation}\label{eq:GiniBound}
    \totdd{G}{t} \leq \gamma ||\rho(\cdot,t)||_\infty \left(1-G(t)\right)
\end{equation} for $0<t<s.$
\end{theorem}

\begin{proof} Let $t\in(0,s).$ Starting from the penultimate line in the main computation of the proof of \cref{thm:GiniIncr}, we have
    \begin{align*}
\totdd{G}{t} 
&= 2\int_0^\infty\,dw\,D[w,\rho]\left(\rho(w)\right)^2 \\
&\leq 2||\rho||_\infty\int_0^\infty\,dw\,D[w,\rho]\rho(w) \\
&=2||\rho||_\infty\int_0^\infty\,dw\,\rho(w)\frac{\gamma}{2}\int_0^\infty\,dx\,\kappa(w,x)\rho(x) &&\text{\quad (definition of $D[w,\rho]$)} \\
&\leq \gamma ||\rho||_\infty\int_0^\infty\,dw\,\rho(w)\int_0^\infty\,dx\,(w\wedge x)\rho(x) &&\text{\quad ($\kappa(w,x)\leq(w\wedge x)$ )}\\
&=  \gamma ||\rho||_\infty(1-G) &&\text{\quad (by \cref{eq:gini})}.
    \end{align*}
\end{proof}

Note that both \cref{thm:GiniIncr} and \cref{thm:GiniBound} can be extended to other binary transactions that modify the original \ysm{} if the transaction kernel meets two conditions.

\begin{corollary}\label{cor:otherModels}
    For a system with transaction kernel $\widetilde{\phi}:\mathbb{R}^+\times \mathbb{R}^+\rightarrow\mathbb{R}^+$ such that \[\widetilde{\phi}\left(w^i,w^j\right) \leq \ysmmin{w^i}{w^j}\] and \[ \widetilde{\phi}\left(w^i,w^j\right)^2 \leq \ysmmin{w^i}{w^j}\] on $\mathbb{R}^+\times \mathbb{R}^+$, we have that \[0\leq\totdd{G}{t}\leq \gamma ||\rho||_\infty (1-G).\]
\end{corollary}

The two conditions on the kernel correspond to not sending agents to negative wealth in the finite model for $\Delta t \in (0,1]$ and the final inequality in the proof of \cref{thm:GiniBound}, respectively. However, the class of kernels satisfying these conditions need not have the stake of a transaction be independent of the wealthier agent. The generalization in \cref{cor:otherModels} is of note as it includes more than just microtransactions with outcomes concentrated via Dirac masses on a few discrete outcomes. Therefore, the class of models for which \cref{thm:GiniBound} holds includes those for which the microtransactions have a continuum of outcomes.

To produce the bound on $G(t)$ we need a slightly modified version of the differential Gr\"onwall's inequality as presented in \cite{LCE2010}.

\begin{lemma}[A version of Gr\"onwall's differential inequality]\label{lem:gronwall}
Let $\eta, \phi,$ and $\psi$ be functions of $t\in[0,T]$. Further let $\eta$ be absolutely continuous and $\phi$ and $\psi$ be integrable on $[0,T]$. If for almost every $t$ \[\dot{\eta}\leq\phi\eta + \psi,\] then \[\eta(t)\leq\exp\left(\int_0^t\,ds\,\phi(s)\right)\left[\eta(0)+\int_0^t\,ds\,\psi(s)\exp\left(-\int_0^s\,dr\,\phi(r)\right)\right].\]
\end{lemma}

\begin{proof}
By direct calculation first and then applying the assumed inequality, we have that \begin{align*}
\totdd{}{s}\left[\eta(s)\exp\left(-\int_0^s\,dr\,\phi(r)\right)\right] &= \exp\left(-\int_0^s\,dr\,\phi(r)\right)\left[\dot{\eta}(s)-\phi(s)\eta(s)\right]\\
&\leq \exp\left(-\int_0^s\,dr\,\phi(r)\right)\psi(s).
\end{align*}
Integrating from $0$ to $t$ in $s$ yields \[\eta(t)\exp\left(-\int_0^t\,dr\,\phi(r)\right)-\eta(0)\leq\int_0^t\,ds\, \psi(s)\exp\left(-\int_0^s\,dr\,\phi(r)\right),\] which upon re-arranging proves the lemma.
\end{proof}
If $\phi$ were assumed to be nonnegative then $\exp\left(-\int_0^s\,dr\,\phi(r)\right)$ is at most unity for all $s$, in which case the upper bound on that term by unity gives the standard inequality.

\begin{claim}
Let \[M_T:=\sup_{t\in[0,T]}||\rho(\cdot,t)||_\infty.\] For $t\in(0,T]$, \[G(t)\leq G(0)\exp\left(-\gamma M_T t\right) + \frac{\exp\left(\gamma M_T t\right)-1}{\exp\left(\gamma M_T t\right) }.\]
\end{claim}

\begin{proof}
Apply \cref{lem:gronwall} with $\eta = G$, $\phi = -\gamma M_T$, and $\psi = \gamma M_T$.
\end{proof}

Thus the magnitude of the exponential rate of convergence to oligarchy under the modified YSM is at most $\gamma M_T$ in a finite time horizon $[0,T]$.

\begin{claim}\label{claim:integratedBound2}
Let $M(t) := ||\rho(\cdot,t)||_\infty$. For $t>0$, \begin{equation}\label{eq:intGiniBound} G(t)\leq G(0)\exp\left(-\gamma\int_0^t\,ds\,M(s)\right) + \frac{\exp\left(\gamma\int_0^t\,ds\,M(s)\right)-1}{\exp\left(\gamma\int_0^t\,ds\,M(s)\right)}.\end{equation}
\end{claim}

\begin{proof}
Apply \cref{lem:gronwall} with $\eta = G$, $\phi = -\gamma M(t)$, and $\psi = \gamma M(t)$. Note that the term \[\int_0^t\,ds\,\psi(s)\exp\left(-\int_0^s\,dr\,\phi(r)\right)\] simplifies to \[\exp\left(\gamma \int_0^t\,ds\,M(s)\right)-1.\]
\end{proof}

The evolutionary, integro-differential equation \cref{eq:mysm} was numerically solved for a random initial condition to demonstrate the bound holding in practice. The diffusion-velocity deterministic particle method of \cite{MR3645392} was used for the time-evolution. We present results with $\gamma = 0.1$, a time step $\Delta t =0.05$, and $N = 10^3$ particles iterated $M = 250$ times. \textit{Mathematica's} built-in kernel density estimate was used to recover the density at each time step from the particle system. The density estimate is necessary for both the evolution of the particles (as their velocities depend on the gradient of the logarithm of the density) and for recording $||\rho(\cdot,t)||_\infty$ at each discrete $t$ for \cref{eq:GiniBound}. The particle positions were evolved using a forward Euler scheme.

\begin{figure}[tbhp]
\centering
\includegraphics[scale = 0.7]{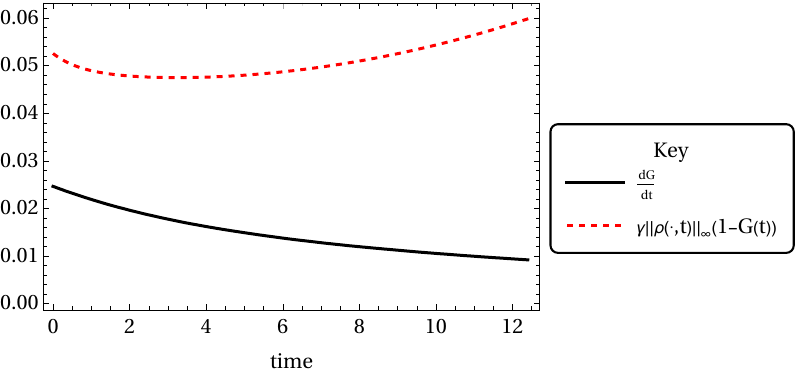}
\caption{Plot of $t\mapsto dG\left[\rho\left(\cdot,t\right)\right]/dt$ and \cref{eq:GiniBound} from a deterministic particle, diffusion-velocity method simulation of \cref{eq:mysm} the modified \ysm{} equation of motion.}
\label{fig:a}
\end{figure}

\begin{figure}[tbhp]
\centering
\includegraphics[scale = 0.75]{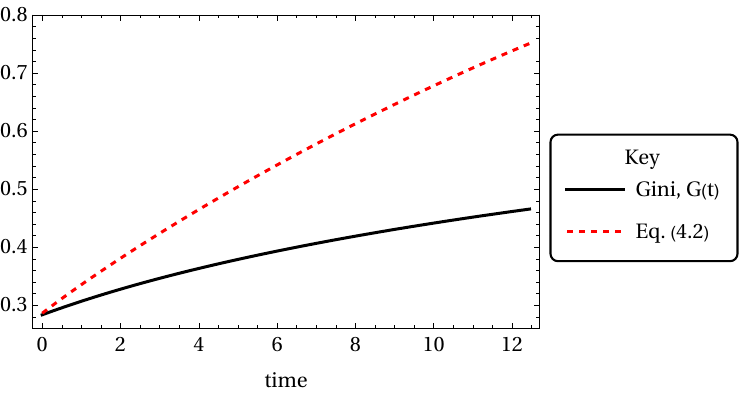}
\caption{Plot of $t\mapsto G\left[\rho\left(\cdot,t\right)\right]$ and \cref{eq:intGiniBound} from a deterministic particle, diffusion-velocity method simulation of \cref{eq:mysm} the modified \ysm{} equation of motion.}
\label{fig:b}
\end{figure}

In both \cref{fig:a,fig:b}, the abscissa is the time variable. In the first figure, the discrete rate of change of the Gini coefficient $\left(G(t+\Delta t) - G(t)\right)/\Delta t$ and the bound $\gamma ||\rho(\cdot, t)||_\infty \left(1-G(t)\right)$ are shown. In the second plot, the Gini coefficient at each time is shown with the integrated version of the bound from \cref{claim:integratedBound2}, using a left Riemann sum approximation to the integral. 

For this experiment, the bound becomes looser as time continues, which is attributable to the accumulation of probability density ever-closer to the origin as oligarchy forms. This accumulation increases $||\rho(\cdot,t)||_\infty$ and indicates that an improved bound would need to do away with the supremum norm of $\rho(\cdot,t)$.

\section{Results of asymptotic analysis}
\label{sec:asymptotics}

\cref{thm:GiniIncr} asserts that inequality monotonically increases under the dynamics of \cref{eq:mysm}, which leads to total wealth condensation. The present state of affairs in the world is not one of total wealth condensation so in order for this model to have explanatory or predictive power, it is necessary to introduce a regularization term and ask if the resulting wealth distributions are close to those of the world.

In \cite{MR3872473}, real distributions of wealth were compared to those arising asymptotically from the YSM with differing combinations of: (1) taxation and redistribution, (2) a bias in the stochastic transaction in favor of the wealthier agent, and (3) negative wealth. It was found that the model's steady state distributions match very well with real wealth distributions. In particular we quote from the summary article \cite{MR4686618} that the YSM, which is a three-parameter model, with the above-mentioned embellishments\begin{quote}\textit{admits steady-state solutions whose Lorenz curves match those of empirical wealth distributions to within one fifth of one percent for American wealth data between 1987 and 2019, and to within a half of a percent for 14 European countries that were associated with the European Central Bank circa 2010.}
\end{quote}

Each of the embellishments could be included in the modified model of \cref{sec:mysm} but the analysis of the resulting system is too far from the main message of this paper: namely, that a Gini -- Gini production bound exists for a non-empty class of unbiased, binary, wealth-conserving kinetic asset exchange models. 

We want to convey to readers that our modification has promise be equally successful in fitting empirical data. To do so, we analyze the steady state of the modified system with just redistribution added and compare it to the classical \ysm{} with redistribution. 

Thus, rather than carrying out comparisons between the redistributive, modified \ysm{} and empirical data from scratch, we instead show that the asymptotics of the redistributive, modified YSM match those of the classical model.

To regularize the modified YSM, we introduce a deterministic taxation and redistribution term that reverts the population to the mean wealth. This scheme can be viewed in two equivalent ways: \begin{itemize}
    \item As an exogenous force that collects the same fraction of each agent's wealth, pools the total, and redistributes the total equally amongst the agents; or
    \item As a partial divorce model so that in addition to each stochastic transaction between agents there is also a deterministic exchange towards their pairwise mean wealth.
\end{itemize} These two perspectives yield the same continuum equation of motion.

The binary transaction for the modified YSM with redistribution is 
\begin{equation}\label{eq:mysmrMicroTransact}
\begin{pmatrix}w_{(k+1)\Delta t}^i\\w_{(k+1)\Delta t}^j\end{pmatrix} = \begin{pmatrix}w_{k\Delta t}^i\\w_{k \Delta t}^j\end{pmatrix} + \left[\chi\Delta t\left(w^j_{k\Delta t}-w^i_{k\Delta t}\right) +\sqrt{\gamma\Delta t} \phi\left(w^i_{k \Delta t},w^j_{k \Delta t}\right)\eta\right] \begin{pmatrix}1\\-1\end{pmatrix},
\end{equation} where $\chi\in[0,\frac{1}{2})$ and the other definitions and parameters from \cref{eq:mysmMicroTransact} and \cref{eq:mysmMicroKernel} carry over. This is no longer a martingale-like process.

The associated equation of motion for the agent density in wealth-space is \begin{equation}\label{eq:mysmr}
\ddt{\rho(w,t)}=-\ddw{}\left[\chi(1-w)\rho(w,t)\right]+\twoddw{}\left[\frac{\gamma}{2}\left(\int_0^\infty\,dx\,\kappa(w,x)\rho(x,t)\right)\rho(w,t)\right],
\end{equation} where $\kappa$ is as before.\footnote{Starting from \cref{eq:mysmr}, both $\gamma$ and $\chi$ lose their initial restrictions imposed from the binary transactions and we demand only their non-negativitity.} The second term on the left hand side is the same diffusion term from \cref{sec:mysm}; this arose from a mean-field assumption that produces the integration against the density. The first term is a deterministic drift term that forces the density to revert to the mean of the initial condition, which we have taken to always be unity. In this way the drift can be viewed as the deterministic part of an Ornstein-Uhlenbeck process. The first-order drift coefficient is easily seen to be the time-infinitesimal expected change in wealth of an agent at wealth $w$, \[\lim_{\Delta t \rightarrow 0}\frac{1}{\Delta t}\mathbb{E}\left[\left(w^i_{(k+1)\Delta t}-w^i_{k\Delta t}\right)\bigg|w^i_{k\Delta t} = w\right],\] where the expectation is taken over both the outcomes of $\eta$ and possible transaction partners.

If $\gamma = 0$ then the stochastic transactions would stop occurring and the deterministic redistribution would occur with rate $\chi$ forcing the density to concentrate on the mean. The case of $\gamma = 0$ and $\chi>0$ corresponds to a minimization of the potential energy functional \[U[\rho] = \int_{\mathbb{R}^+}\,dw\,\left(\frac{w^2}{2}-w\right)\rho(w)\] in the long time limit.

In this section, we closely follow the procedure of asymptotic analysis laid out in \cite{BMB2017}. In what follows, $\rho_\infty(w)$ is the asymptotic state of \cref{eq:mysmr} for which the time derivative is set to zero and the ODE is studied. The asymptotic state $\rho_\infty(w)$ may be interpreted as the distribution of wealth such that transactions still occur across the entire population and in the finite-agent model each agent would still experience changes in their own wealth, but the population-wide statistics no longer change. That is to say, as individual fortunes change, any sampling and measurement procedure would yield the same outcome once the population has arrived at the invariant measure.

Let us emphasize that the reason for comparing $\rho_\infty$ of the redistributive versions of the modified \ysm{} and classical \ysm{} is that the previous work of \cite{MR3872473} had success with fitting the steady-states to empirical data.

The details of the analysis are in \cref{sec:asymptoticsDet}. At $w\ll 1$, \[\rho_\infty(w) = \frac{c_0}{w^{2+2\chi/\gamma}}\exp\left(-\frac{2\chi}{\gamma w}\right),\] where $c_0>0$. This small $w$ behavior has the same characteristic shape as the classic \ysm{} with redistribution \cite{BMB2017}. 

The more involved computation is for $w\gg 1$ for which $\rho_\infty$ is found to be Gaussian, \begin{equation}
    \rho_\infty(w)\approx c_\infty \exp(-aw^2-bw)
\end{equation} for $a>0$, $b\in \mathbb{R}$, and $c_\infty>0$. The constants are defined in \cref{sec:asymptoticsDet}.

The classical \ysm{} with redistribution also has a Gaussian tail. Hence, for both $w\ll 1$ and $w\gg 1$, the redistributive versions of the modified and classical \ysm{} have the same characteristic shapes. Since the latter has successfully been used to fit empirical data, we take these results to indicate that the modified system holds similar promise.

\section{Conclusions}
\label{sec:conclusions}
We introduced a variant of the \ysm{} for which the Gini coefficient of economic inequality monotonically increases under the resulting continuum dynamics yet the rate of change in time of the Gini coefficient permits an upper bound. The way in which this bound holds is similar to the entropy -- entropy production bounds for nonlinear Fokker-Planck equations. In the econophysics case, the twin results of \cref{cor:strictlyIncr} and \cref{thm:GiniBound} may be interpreted as the adage wealth begets wealth but with the constraint that the accumulation of wealth into a small portion of society begins to limit how quickly more can be extracted from the poor.

The method of bounding $\dot{G}$ used techniques from the analysis of deterministic equation of motions. It is not clear if or how this bound could be applied to the stochastic, finite-agent models that precede the diffusion approximation. Nor is it clear how a similar bound, phrased in a probabilistic way, could be derived directly from the stochastic, finite-agent system.

Since the bound in \cref{thm:GiniBound} includes $||\rho||_\infty$, the $L^\infty$ norm of the density, it is natural to ask how the bound could be tested with empirical data since the transition from empirical data to a density function involves user choices (e.g., width in kernel density estimation) that affect the essential supremum of the resulting density. Thus we consider the development of a method that examines whether real-world economic data obey this bound to be an open question.

By showing that the asymptotics of this modified model with redistribution match that of the original \ysm{} with redistribution, we have put forward an argument that the macroscopic asymptotics are robust to changes of the microscopic transactions. Thus many microscopic transactions (and not \textit{just} a practitioner's favorite) may give rise to simple yet accurate descriptions of the evolution of wealth that depend on dramatically fewer parameters than most economic theories.

Finally, the Gini -- Gini production inequality of \cref{thm:GiniBound} was shown to hold not only for the particular modification \cref{eq:mysm} but also for a broader class of models described in \cref{cor:otherModels}. In doing so, we can consider a class of unbiased, wealth-conserving kinetic asset exchange models, which produce non-linear, integro-differential evolution equations of the McKean-Vlasov type, that possess a ``thermodynamical'' second law of Gini coefficient monotonicity and obey the Gini -- Gini production inequality, \[0<\totdd{}{t}\bigg|_{t=s}\left(G\circ \rho_t\right) \leq \gamma ||\rho_s||_\infty (1-G[\rho_s]).\]

\appendix
\section{Details of asymptotic analysis}
\label{sec:asymptoticsDet}

\subsection{Analysis for equilibrium at \texorpdfstring{{\boldmath$w\ll 1$}}{w much less than 1}}
For $w\ll 1$, $\kappa(w,x) = \ysmmin{w}{x}^2$ thus \[D[w,\rho] = \frac{\gamma}{2}\left(\int_0^w\,dx\,x^2\rho(x,t)+w^2\int_w^\infty\,dx\,\rho(x,t)\right).\]
We approximate this by $D[w,\rho]\approx \frac{\gamma}{2}w^2$. This assumption demands the same \textit{a posteriori} justification as noted in \cite{BMB2017}. At equilibrium \[\chi(1-w)\rho_{\infty}(w) = \totdd{}{w}\left[\frac{\gamma}{2}w^2\rho_\infty(w)\right].\]

This is solved by \[\rho_\infty(w) = \frac{c_0}{w^{2+2\chi/\gamma}}\exp\left(-\frac{2\chi}{\gamma w}\right),\] where $c_0$ is a positive constant. This agrees with the results in \cite{BMB2017} where this analysis was carried out for the classic \ysm{} with redistribution. Thus the justification for the earlier assumption about the behavior of $D[w,\rho]$ at $w\ll 1$ is equally valid as in the original paper.

\subsection{Analysis for equilibrium at \texorpdfstring{{\boldmath$w\gg 1$}}{w much greater than 1}}
Upon integrating once and re-arranging, the large-$w$ equilibrium condition is \begin{equation}\label{eq:largeW}
\totdd{\log \rho_\infty(w)}{w} = \frac{\chi(1-w)-\totdd{D[w,\rho_\infty]}{w}}{D[w,\rho_\infty]}.
\end{equation}

We first investigate the behavior of $D[w,\rho]$ for $w\gg 1$. The piecewise conditions of $\kappa(w,x)$ partially lose their dependence on $w$ so that \begin{align*}D[w,\rho]&=\frac{\gamma}{2}\int_0^\infty\,dx\,\rho(x)\begin{cases}
x^2 &\text{\quad if $ x <1$;} \\
(w \wedge x)&\text{\quad if $ x \geq 1$,}\end{cases}\\
&=\frac{\gamma}{2}\left(\int_0^1\,dx\,x^2\rho(x)+\int_1^w\,dx\,x\rho(x)+w\int_w^\infty\,dx\,\rho(x)\right).
\end{align*}

Let \[c_1:=\int_0^1\,dx\,x^2\rho(x)\] and \[c_2 := \int_0^\infty\,dx\,x\rho(x)-\int_0^1\,dx\,x\rho(x),\] noting that $c_2>0$. Thus \begin{equation}\label{eq:DwithConstants}
D[w,\rho] = \frac{\gamma}{2}\left(c_1+c_2-\int_w^\infty\,dx\,x\rho(x)+w\int_w^\infty\,dx\,\rho(x)\right).
\end{equation}

Therefore we also have \begin{equation}\label{eq:largeWdDdw}
\totdd{D[w,\rho]}{w} =\frac{\gamma}{2}\int_w^\infty\,dx\,\rho(x).
\end{equation}

We make the ansatz that \begin{equation}\label{eq:largeWAnsatz}
    \rho_\infty(w)\approx c_\infty \exp(-aw^2-bw)
\end{equation} for $w\gg1$, $a>0$, $b\in \mathbb{R}$, and $c_\infty>0$. In this case, \begin{equation}
    \label{eq:logLargeW}
    \totdd{\log \rho_\infty(w)}{w} = -2aw -b.
\end{equation}

We simplify the right hand side of \cref{eq:largeW} term-by-term under the ansatz \cref{eq:largeWAnsatz} to check whether it is an affine function of $w$, and if so identify the appropriate constants. 

In doing so, we make repeated use of the asymptotic approximation \[\text{erfc}(z) \approx \frac{1}{z\sqrt{\pi}}\exp(-z^2).\]

Starting first with \cref{eq:largeWdDdw}
\begin{align}\int_w^\infty\,dx\,\rho_\infty(x) &= \frac{c_\infty}{2}\sqrt{\frac{\pi}{a}}\exp\left(\frac{b^2}{4a}\right)\text{erfc}\left(\frac{2aw+b}{2\sqrt{a}}\right) \nonumber \\
&\approx \frac{c_\infty}{2}\sqrt{\frac{\pi}{a}}\exp\left(\frac{b^2}{4a}\right) \frac{2\sqrt{a}}{2aw+b}\frac{1}{\sqrt{\pi}}\exp\left(-\left(\frac{2aw+b}{2\sqrt{a}}\right)^2\right) \nonumber \\
&=\frac{c_\infty}{2aw+b}\exp(-aw^2-bw).\label{eq:asympAnalysis1} 
\end{align}

Next we turn to the terms with $w$ dependence in \cref{eq:DwithConstants}. The first of which is 
\begin{align}\int_w^\infty\,dx\,x\rho_\infty(x) &= c_\infty\frac{\exp\left(-aw^2-bw\right)}{2a}-\frac{b}{2a}\left[\frac{c_\infty}{2}\sqrt{\frac{\pi}{a}}\exp\left(\frac{b^2}{4a}\right)\text{erfc}\left(\frac{2aw+b}{2\sqrt{a}}\right)\right] \nonumber \\
&\approx c_\infty\frac{\exp\left(-aw^2-bw\right)}{2a}-\frac{b}{2a}\left[\frac{c_\infty}{2aw+b}\exp(-aw^2-bw)\right] \nonumber \\
&=\frac{c_\infty}{2a}\exp(-aw^2-bw)\left(1-\frac{b}{2aw+b}\right) \nonumber \\
&=w\frac{c_\infty}{2aw+b}\exp(-aw^2-bw).\label{eq:asympAnalysis2}
\end{align}

Combining \cref{eq:asympAnalysis1,eq:asympAnalysis2}, we have that \[-\int_w^\infty\,dx\,x\rho_\infty(x)+w\int_w^\infty\,dx\,\rho_\infty(x) \approx 0,\] which implies $D[w,\rho_\infty] \approx \frac{\gamma}{2}(c_1+c_2).$

Taking these results together, we see that upon using \cref{eq:largeWAnsatz} and approximating asymptotically, \cref{eq:largeW} simplifies to \[\totdd{\log \rho_\infty(w)}{w}=\frac{2\chi}{\gamma(c_1+c_2)}(1-w)-\underbrace{\frac{c_\infty}{c_1+c_2}\frac{\exp\left(-aw^2-bw\right)}{2aw+b}}_{\text{subdominant corrections}}.\]

In particular, we note from \cref{eq:logLargeW} that $a = \frac{\chi}{\gamma(c_1+c_2)}>0$.

Thus the characteristic shape of the time-asymptotic wealth profile is non-analytic and depleted at the origin and Gaussian at large $w$. This is the same profile found for the classic \ysm{} with redistribution.

\section*{Acknowledgments}
We thank M. Johnson and D. Gentile of Tufts University for productive conversations.

\bibliographystyle{siamplain}
\bibliography{references}
\end{document}